\documentclass{amsart}
\usepackage{amsmath}
\usepackage{amsthm}
\usepackage{amsfonts}
\usepackage{amssymb}
\usepackage[utf8]{inputenc}
\usepackage{tikz}

\newcommand{\gv}[1]{\mathbf{#1}}  % graph variable, i.e. G-indexed variable

\newcommand{\subs}{\subseteq}

\newcommand{\cut}{\cap}
\newcommand{\uni}{\cup}

\newcommand{\Or}{\vee}
\renewcommand{\And}{\wedge}

\newcommand{\Sp}[2]{\{{#1}\ |\ {#2}\}} % sets  by property

\newcommand{\gs}[1]{|{#1}|}            % size of a graph

\newcommand{\Ss}[1]{\|{#1}\|}                   % sum of a set

\newcommand{\sd}{\bigtriangleup}          % symmetric difference of sets

\newcommand{\eps}{\varepsilon}
\renewcommand{\phi}{\varphi}

\newcommand{\prc}[1]{\ensuremath{\mathsf{#1}}}
\newcommand{\SP}{\prc{\#P}}

\newcommand{\NP}{\prc{NP}}

\newcommand{\poly}{\mathsf{poly}}

\newcommand{\ETH}{\#ETH}

\newcommand{\DRSAT}{\#3-SAT}
\newcommand{\XDRSAT}{\#X3SAT}
\newcommand{\nIS}{\#$\frac 1 3$-IS}

\newcommand{\ProbName}{\item[Name:]}
\newcommand{\ProbInput}{\item[Input:]}
\newcommand{\ProbOutput}{\item[Output:]}
\newenvironment{problem}
{\begin{list}{*}{\setlength\itemsep{0cm}
\setlength\parsep{0cm}}}
{\end{list}}

\newcommand{\gpe}[2]{{#1}(-;{#2})}       % graph polynomial evaluation
\newcommand{\Q}{\mathbb{Q}}

\newcommand{\R}{\mathbb{R}}

\renewcommand{\S}{\tilde S}
\newcommand{\T}{\tilde T}
\newcommand{\Sot}{S_{12}}
\newcommand{\Sz}{\tilde S_0}

\newcommand{\lo}{{\lambda_1}}
\newcommand{\lt}{{\lambda_2}}

\newcommand{\vdel}[2]{{#1}-{#2}}   % vertex deletion

\newtheorem{defi}{Definition}[section]
\newtheorem{cor}[defi]{Corollary}
\newtheorem{lem}[defi]{Lemma}
\newtheorem{thm}[defi]{Theorem}
\newtheorem{rem}[defi]{Remark}

\newtheorem*{expth}{\ETH}

\newcommand{\defexpr}[1]{{\em{#1}}}     % Definition eines Begriffes

\newcommand{\tr}[2][]{\tau_{{#1}} {#2}}     % general (graph) transformation

\newcommand{\tc}[1][]{\kappa^{{#1}}}          % graph transformation ``clone''

\begin{document}

\title[Exponential Time Complexity of Weighted Counting Independent
  Sets]{Exponential Time Complexity of \\ Weighted Counting of
  Independent Sets}

\author{Christian Hoffmann}

\email{christian.hoffmann2010@googlemail.com}
\thanks{Full version of a contribution to IPEC 2010. This work has been done while the author was a research assistant at Saarland University, Germany.}

\begin{abstract}
  We consider weighted counting of independent sets using a rational
  weight $x$: Given a graph with $n$ vertices, count its independent
  sets such that each set of size $k$ contributes $x^k$. This is
  equivalent to computation of the partition function of the lattice
  gas with hard-core self-repulsion and hard-core pair interaction. We
  show the following conditional lower bounds: If counting the
  satisfying assignments of a $3$-CNF formula in $n$ variables
  (\#3SAT) needs time $2^{\Omega(n)}$ (i.e.\ there is a $c>0$ such
  that no algorithm can solve \#3SAT in time $2^{cn}$), counting the
  independent sets of size $n/3$ of an $n$-vertex graph needs time
  $2^{\Omega(n)}$ and weighted counting of independent sets needs time
  $2^{\Omega(n/\log^3 n)}$ for all rational weights $x\neq 0$.

  We have two technical ingredients: The first is a reduction from
  3SAT to independent sets that preserves the number of solutions and
  increases the instance size only by a constant factor. Second, we
  devise a combination of vertex cloning and path addition. This graph
  transformation allows us to adapt a recent technique by Dell,
  Husfeldt, and Wahlén which enables interpolation by a family of
  reductions, each of which increases the instance size only
  polylogarithmically.
\end{abstract}

\maketitle

\section{Introduction}

Finding independent sets with respect to certain restrictions is a
fundamental problem in theoretical computer science. Perhaps the most
studied version is the maximum independent set problem: Given a graph,
find an independent set\footnote{A subset $A$ of the vertices of a
  graph is \defexpr{independent} iff no two vertices in $A$ are joined
  by an edge of the graph.} of maximum size. This problem is closely
related to the clique and vertex cover problems. The decision versions
of these are among the 21 problems considered by Karp in 1972
\cite{karp}, and they are used as examples in virtually every
exposition of the theory of $\NP$-completeness \cite[Section
  3.1]{garey_johnson}, \cite[Section 9.3]{papadimitriou},
\cite[Section 34.5]{cormen}. Exact algorithms for the independent set
problem have been studied since the 70s of the last century
\cite{DBLP:journals/siamcomp/TarjanT77, DBLP:journals/tc/Jian86,
  DBLP:journals/jal/Robson86} and there is still active research
\cite{DBLP:journals/jacm/FominGK09, kneis_et_al:LIPIcs:2009:2326,
  DBLP:conf/swat/BourgeoisEPR10}.

Besides finding a maximum independent set, algorithms that count the
number of independent sets have also been developed
\cite{DBLP:conf/soda/DahllofJ02}. If the counting process is done in a
weighted manner (as in \eqref{eq:indset_poly} below), we arrive at a
problem from statistical physics: computation of the partition
function of the lattice gas with hard-core self-repulsion and
hard-core pair interaction \cite{scott_sokal}. In graph theoretic
language, this is the following problem: Given a graph $G=(V,E)$ and a
weight $x\in \Q$, compute
\begin{equation}
  \label{eq:indset_poly}
  I(G;x)=\sum_{\substack{A\subs V\\ \text{$A$ independent set}}} x^{|A|}.
\end{equation}
$I(G;x)$ is also known as the independent set polynomial of $G$
\cite{HoedeLi, GutmanHarary}. Luby and Vigoda mention that
``equivalent models arise in the Operations Research community when
considering properties of stochastic loss systems which model
communication networks'' \cite{luby_vigoda}. Evaluation of $I(G;x)$
has received a considerable amount of attention, mainly concerning
approximability if $x$ belongs to a certain range depending on
$\Delta$, the maximum degree of $G$ \cite{luby_vigoda, dyer00markov,
  DBLP:journals/combinatorics/Vigoda01, DBLP:conf/stoc/Weitz06}.

In this paper, we give evidence that exact evaluation of $I(G;x)$
needs almost exponential time (Theorem~\ref{thm:WISexponential}). We
do this by reductions from the following problem:

\begin{problem}
\ProbName \#$d$-SAT
\ProbInput Boolean formula $\phi$ in $d$-CNF with $m$ clauses in $n$ variables
\ProbOutput Number of satisfying assignments for $\phi$
\end{problem}

All lower bounds of this work are based on the following assumption,
which is a counting version of the exponential time hypothesis (ETH)
\cite{impagliazzo_paturi_zane, dell_husfeldt_wahlen}:

\begin{expth}[Dell, Husfeldt, Wahlén 2010]
\label{conj:eth}
There is a constant $c > 0$ such that no deterministic algorithm can
compute \DRSAT\ in time $\exp(c \cdot n)$.
\end{expth}

Our first result concerns the following problem:
\begin{problem}
\ProbName \nIS
\ProbInput Graph with $n$ vertices
\ProbOutput Number of independent sets of size exactly $n/3$ in $G$
\end{problem}

\begin{thm}
  \label{thm:IS}
  \nIS\ requires time $\exp(\Omega(n))$ unless \ETH\ fails.
\end{thm}

Theorem~\ref{thm:IS} gives an important insight for the development of
exact algorithms counting independent sets: Let us consider algorithms
that count independent sets of a particular kind. (For example:
algorithms that count the independent sets of maximum size. Another
example: algorithms that simply count all independent sets). Using
only slight modifications, many of the actual algorithms that have
been suggested for these problems can be turned into algorithms that
solve \nIS. Theorem~\ref{thm:IS} tells us that there is some $c>1$
such that every such algorithm has worst-case running time $\geq
c^n$---unless \ETH\ fails. In other words: There is a universal $c^n$
barrier for counting independent sets that can only be broken 1) if
substantial progress on counting SAT is made or 2) by approaches that
are custom-tailored to the actual version of the independent set
problem such that they can not be used to solve \nIS.

The proof of Theorem~\ref{thm:IS} is \emph{different} from the
standard constructions that reduce the decision version of 3SAT to the
decision version of maximum independent set \cite{karp}, \cite[Theorem
  9.4]{papadimitriou}. This is due to the fact that these
constructions do not preserve the number of solutions. Furthermore,
the arguments for counting problems that have been applied in
\SP-hardness proofs also fail in our context, as they increase the
instance size by more than a constant factor\footnote{For instance,
  Valiant's step from perfect matchings to prime implicants
  \cite{valiant_enumeration} includes transforming a $\Theta(n)$
  vertex graph into a $\Theta(n^2)$ vertex graph.} and thus do not
preserve subexponential time.

Theorem~\ref{thm:IS} is proved in Section~\ref{sec:sat2IS} using a,
with hindsight, simple reduction from \DRSAT. But for the reasons
given in the last paragraph, it is important to work this out
precisely. In this way, we close a non-trivial gap to a result that is
very important as it concerns a fundamental problem.

The main result of our paper is based on Theorem~\ref{thm:IS}:
\begin{thm}
\label{thm:WISexponential}
  Let $x\in \Q$, $x\neq 0$. On input $G=(V,E)$, $n=|V|$, evaluating
  the independent set polynomial at $x$, i.e.\ computing
  \[ \sum_{\substack{A\subs V\\\text{$A$ independent set}}} x^{|A|}, \]
  requires time $\exp(\Omega(n/\log^3n))$ unless \ETH\ fails.
\end{thm}
This shows that we can not expect that the partition function of the
lattice gas with hard-core self-repulsion and hard-core pair
interaction can be computed much faster than in exponential time.

Let us state an important consequence of
Theorem~\ref{thm:WISexponential}, the case $x=1$.
\begin{cor}
  \label{cor:IS}
  Every algorithm that, given a graph $G$ with $n$ vertices, counts
  the independent sets of $G$ has worst-case running time
  $\exp(\Omega(n/\log^3 n))$ unless \ETH\ fails.
\end{cor}

Referring to the discussion after Theorem~\ref{thm:IS}, this gives a
conditional lower bound for our second example (i.e.\ counting all
independent sets of a graph). The bound of Corollary~\ref{cor:IS} is
not as strong as the bound of Theorem~\ref{thm:IS} but holds for
\emph{every} algorithm, not only for algorithms that can be modified
to solve \nIS.

\subsection*{Techniques and Relation to Previous Work}

Theorem~\ref{thm:IS} is proved in two steps: First, we reduce from
\DRSAT\ to \XDRSAT. \XDRSAT\ is a version of SAT that counts the
assignments that satisfy \emph{exactly} one literal per clause. From
\XDRSAT\ we can reduce to independent sets using a modified version of
a standard reduction from SAT to independent sets \cite[Theorem
  9.4]{papadimitriou}. We also use the fact that the exponential time
hypothesis with number of variables as a parameter is equivalent to
the hypothesis with number of clauses as parameter. Impagliazzo,
Paturi, and Zane proved this for the decision version
\cite{impagliazzo_paturi_zane}. We use the following version for
counting problems:

\begin{thm}[{\cite[Theorem 1]{dell_husfeldt_wahlen}}]
\label{thm:eth_variables_vs_clauses}
  For all $d\geq 3$, \ETH\ holds if and only if \#$d$-SAT requires
  time $exp(\Omega(m))$.
\end{thm}

Our main result (Theorem~\ref{thm:WISexponential}) is inspired by
recent work of Dell, Husfeldt, and Wahlén on the Tutte polynomial
\cite[Theorem 3(ii)]{dell_husfeldt_wahlen}. These authors use Sokal's
formula for the Tutte polynomial of generalized Theta graphs
\cite{sokal_chromatic}. In Section~\ref{sec:cloning}, we devise and
analyze $S$-clones, a new graph transformation that can be used with
the independent set polynomial in a similar way as generalized Theta
graphs with the Tutte polynomial. $S$-clones are a combination of
vertex cloning (used under this name for the interlace polynomial
\cite{interlace_hard}, but generally used in different situations for
a long time \cite[Theorem 1, Reduction 3.]{valiant_enumeration},
\cite{jerrum_valiant_vazirani}, \cite[Lemma A.3]{roth}) and addition
of paths. Having introduced $S$-clones, we are able to transfer the
construction of Dell et al.\ quite directly to the independent set
polynomial. The technical details are more involved than in the
previous work on the Tutte polynomial, but the general idea is the
same: Use the graph transformation ($S$-clones) to evaluate the graph
polynomial (independent set polynomial) at different points, and use
the result for interpolation. An important property of the
construction is that the graph transformation increases the size of
the graph only polylogarithmically. More details on this can be found
at the beginning of Section~\ref{sec:interpolation}.

Before we start with the detailed exposition, let us mention that the
reductions we devise for the independent set polynomial can be used
with the interlace polynomial \cite{arratia_two_var_interl,
  courcelle_interlace_final} as well \cite{hoffmann}.

%%%%%%%%%%%%%%%%%%%%%%%%%%%%%%%%%%%%%%%%%%%%%%%%%%%%%%%%%%%%%%%%%%%%%%
\section{Reduction from counting SAT to counting independent sets}
\label{sec:sat2IS}

We give the details of a reduction from SAT to independent sets which
increases the instance size only by a constant factor and preserves
the number of solutions. This yields the conditional lower bound for
counting independent sets of size $n/3$ (Theorem~\ref{thm:IS}).

\begin{problem}
\ProbName \XDRSAT
\ProbInput Boolean formula $\phi=C_1\And\ldots \And C_m$ where each
clause $C_i$ is a disjunction of two or three literals over variables
$x_1, \ldots, x_n$
\ProbOutput Number of assignments for $\phi$ such that in every clause exactly one literal is satisfied
\end{problem}

By a \defexpr{polynomial time reduction} from a counting problem $A$
to a counting problem $B$ we mean a polynomial time algorithm that
maps an input instance $x$ for $A$ to an input instance $y$ for $B$
such that the number of solutions for $x$ equals the number of
solutions for $y$.

\begin{lem}
  \label{lem:sat2xsat}
  There is a polynomial time reduction from \DRSAT\ to \XDRSAT\ that
  maps formulas with $m$ clauses to formulas with $O(m)$ clauses.
\end{lem}
\begin{proof}
  Schaefer \cite[Lemma 3.5]{DBLP:conf/stoc/Schaefer78} gives the
  following construction. For a clause $C=(a\Or b \Or c)$, define
  $F=(a\Or u_1 \Or u_4)(b\Or
  u_2\Or u_4)(u_1\Or u_2\Or u_5)(u_3\Or u_4\Or u_6)(c\Or u_3)$.

  It is not hard to check that every assignment that satisfies $C$ in
  the usual sense corresponds to exactly one assignment that satisfies
  $F$ in the sense of \XDRSAT.

  Using this construction, we reduce \DRSAT\ to \XDRSAT. If an
  instance $\phi$ for \DRSAT\ is given, we construct an instance
  $\phi'$ for \XDRSAT\ by applying the above construction for every
  clause in $\phi$, each time using ``fresh'' variables $u_1, \ldots,
  u_6$. 
\end{proof}

\begin{lem}
  \label{lem:ISexplowerbound}
  There is a polynomial time reduction from \XDRSAT\ to \nIS\ that maps
  formulas with $m$ clauses to graphs with $3m$ vertices.
\end{lem}
\begin{proof}
  We reduce from \XDRSAT. Let $\phi=C_1\And \ldots \And C_m$ be an
  instance for \XDRSAT, i.e.\ a 3-CNF formula with $m$ clauses in $n$
  variables. We assume that every variable appears in $\phi$,
  otherwise a factor of $2^r$ is introduced in the following
  reduction, where $r$ is the number of variables that do not appear
  in $\phi$. Furthermore, we assume that no literal appears twice in a
  clause and that, if a literal $\ell$ appears in a clause, its
  negation $\neg \ell$ does not appear in the same clause. The
  construction in Lemma~\ref{lem:sat2xsat} complies with these
  assumptions. Therefore, we do not lose generality.

  For each clause $C_i=\ell_{i,1} \Or \ell_{i,2}\Or \ell_{i,3}$ of
  $\phi$, we construct a triangle $T_i$ whose vertices
  $v_{i,1},v_{i,2}, v_{i,3}$ are labeled by
  $\ell(v_{i,j})=\ell_{i,j}$, $1\leq j\leq 3$, the literals of
  $C_i$. In this way, we obtain the vertex set $V=\{v_{i,j}\ |\ 1\leq
  i\leq m, 1\leq j\leq 3\}$ for the \nIS\ instance $G$. Besides the
  triangle edges, we add the following edges to $G$: For each pair
  $\{u,v\}$ of vertices, where $\ell(u)=\ell(v)$ or $\ell(u) = \neg
  \ell(v)$, let $u_2, u_3$ be the other two vertices in $u$'s triangle
  and $v_2, v_3$ be the other two vertices in $v$'s triangle. If
  $\ell(u)=\ell(v)$, we connect $u$ to $v_2$ and $v_3$, and we connect
  $v$ to $u_2$ and $u_3$. If $\ell(u)=\neg \ell(v)$, we connect $u$
  and $v$, and we connect every vertex of $\{u_2, u_3\}$ to every
  vertex of $\{v_2, v_3\}$.

  It is not difficult to argue that the number of satisfying
  assignment for $\phi$ (i.e.\ assignments such that in each clause
  exactly on literal evaluates to true) equals the number of
  independent sets $A$ of $G$ with $|A|=m$.
\end{proof}

\begin{proof}[Proof of Theorem~\ref{thm:IS}]
  Follows from Theorem~\ref{thm:eth_variables_vs_clauses},
  Lemma~\ref{lem:sat2xsat}, and Lemma~\ref{lem:ISexplowerbound}.
\end{proof}

%%%%%%%%%%%%%%%%%%%%%%%%%%%%%%%%%%%%%%%%%%%%%%%%%%%%%%%%%%%%%%%%%%%%%%
\section{$S$-clones and the Independent Set Polynomial}
\label{sec:cloning}

In this section, we analyze the effect of the following graph
transformation on the independent set polynomial.

\begin{defi}
  Let $S$ be a finite multiset of nonnegative integers and $G=(V,E)$
  be a graph. We define the \defexpr{$S$-clone} $G_S=(V_S,E_S)$ of $G$
  as follows:
  \begin{itemize}
  \item For every vertex $a\in V$, there are $|S|$ vertices
    $a(|S|):=\{a_1, \ldots, a_{|S|}\}$ in $V_S$.
  \item For every edge $uv\in E$, there are edges in $E_S$ that
    connect every edge in $u(|S|)$ to every edge in $v(|S|)$.
  \item Let $S=\{s_1, \ldots, s_\ell\}$. For every vertex $a\in V$, we
    add a path of length $s_i$ to $a_i$, the $i$th clone of
    $a$. Formally: For every $i$, $1\leq i\leq |S|$, and every $a\in
    V$ there are vertices $a_{i,1}, \ldots, a_{i,s_i}$ in $V_S$ and
    edges $a_ia_{i,1}, a_{i,1}a_{i,2}, \ldots, a_{i,s_i-1}a_{i,s_i}$
    in $E_S$.
  \item There are no other vertices and no other edges in $G_S$ but
    the ones defined by the preceding conditions.
  \end{itemize}
\end{defi}

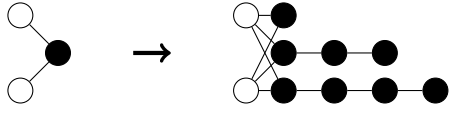
\begin{figure}
  \begin{center}
    \tikzstyle{knoten}=[draw,circle]
    \tikzstyle{kn2}=[knoten,fill]
    \begin{tikzpicture}
      \node (u0) at (0,1) [knoten] {};
      \node (u2) at (0,0) [knoten] {};
      \node (u1) at (0.5,0.5) [kn2] {};
      \draw (u0)--(u1)--(u2);

      \draw [->,ultra thick] (1.5,0.5)--(2,0.5);

      \node (v0) at (3,1) [knoten] {};
      \node (v2) at (3,0) [knoten] {};
      \node (v10) at (3.5,1) [kn2] {};
      \node (v11) at (3.5,0.5) [kn2] {};
      \node (v12) at (3.5,0) [kn2] {};
      \draw (v0)--(v10)--(v2);
      \draw (v0)--(v11)--(v2);
      \draw (v0)--(v12)--(v2);
      \node[kn2] (v111) at (v11) [right=0.5cm] {};
      \node[kn2] (v112) at (v111) [right=0.5cm] {};
      \draw (v11)--(v111)--(v112);
      \node[kn2] (v121) at (v12) [right=0.5cm] {};
      \node[kn2] (v122) at (v121) [right=0.5cm] {};
      \node[kn2] (v123) at (v122) [right=0.5cm] {};
      \draw (v12)--(v121)--(v122)--(v123);
    \end{tikzpicture}
  \end{center}
  \caption{Effect of a $\{0,2,3\}$-clone on a single vertex.}
  \label{fig:Sclone}
\end{figure}

The effect of $S$-cloning on a single vertex is illustrated in
Figure~\ref{fig:Sclone}. The purpose of $S$-clones is that $I(G_S;x)$
can be expressed in terms of $I(G;x(S))$, where $x(S)$ is some number
derived from $x$ and $S$. For technical reasons, we restrict ourselves
to $x$ that fulfill the following condition:

\begin{defi}
  Let $x\in \R$. We say that $x$ is \defexpr{nondegenerate for path
    reduction} if $x>-\frac 1 4$ and $x\neq 0$. Otherwise, we say that
  $x$ is degenerate for path reduction.
\end{defi}

\begin{defi}
  \label{def:lambda}
  Let $x\in \R$ be nondegenerate for path reduction. Then we define
  $\lo$ and $\lt$ to be the two roots of
  \begin{equation}
    \label{eq:lambda}
    \lambda^2-\lambda-x,
  \end{equation}
  i.e.
  \begin{equation}
    \label{eq:lolt}
    \lambda_{1,2}=\frac 1 2 \pm \sqrt{\frac 1 4 + x}.    
  \end{equation}
\end{defi}

The following condition ensures that \eqref{eq:xS} is well-defined
(cf.\ \eqref{eq:Ck}).

\begin{defi}
  Let $x\in \R$ be nondegenerate for path reduction. We say that a set
  $S$ of nonnegative integers is \defexpr{compatible} with $x$ if
  $\lo^{s+2}\neq\lt^{s+2}$ for all $s\in S$.
\end{defi}

Now we can state the effect of $S$-cloning on the independent set
polynomial:

\begin{thm}
  \label{thm:S_cloning}
  Let $G=(V,E)$ be a graph, $x$ be nondegenerate for path reduction,
  and $S$ be a finite multiset of nonnegative integers that is
  compatible with $x$. Then we have
  \[ I(G_S;x) = (\prod_{s\in S}C_s)^{|V|}I(G;x(S)), \]
  where 
  \begin{eqnarray}
    \label{eq:xS}
    x(S)+1&=&\prod_{s\in S}\Big(1+\frac {B_s}{C_s}\Big)\quad \text{with} \\
    \label{eq:Bk}
    B_k &=& \frac x {\lambda_2-\lambda_1} \cdot \big(-\lambda_1^{k+1}+\lambda_2^{k+1}\big), \\
    \label{eq:Ck}
    C_k &=& \frac 1 {\lambda_2-\lambda_1} \cdot \big(-\lambda_1^{k+2} + \lambda_2^{k+2}\big),
  \end{eqnarray}
  and $\lo$, $\lt$ as in Definition~\ref{def:lambda}.
\end{thm}

The rest of this section is devoted to the proof of
Theorem~\ref{thm:S_cloning}. 

\subsection{Notation}

We use a \emph{multivariate} version of the independent set
polynomial.  This means that every vertex has its own variable
$x$. Formally, we define a \defexpr{vertex-indexed} variable $\gv x$
to be a set of of independent variables $x_a$ such that, if $G=(V,E)$
is a graph, $\gv x$ contains $\Sp {x_a} {a\in V}$. If $\gv x$ is a
vertex-indexed variable and $A$ is a subset of the vertices of $G$, we
define
\[ x_A := \prod_{a\in A} x_a.\]
The \defexpr{multivariate independent
  set polynomial} \cite{scott_sokal} is defined as
\begin{equation}
\label{eq:indset_mvar}
  I(G;\gv x)= \sum_{\substack{A\subs V\\A\ \text{independent}}} x_A.
\end{equation}
We have $I(G;x)=I(G;\gv x)[x_a:=x\ |\ a\in V]$, i.e.\ the
single-variable independent set polynomial is obtained from the
multivariate version by substituting every vertex-indexed variable
$x_a$ by one and the same ordinary variable $x$.

We will use the following operation on graphs: Given a graph $G$ and a
vertex $b$ of $G$, $\vdel G b$ denotes the graph that is obtained from
$G$ by removing $b$ and all edges incident to $b$. \label{def:vdel}

\subsection{Proof of Theorem~\ref{thm:S_cloning}}

Let us first analyze the effect of a
single leaf on the independent set polynomial.

\begin{lem}
  \label{lem:deg1red}
  Let $G=(V,E)$ be a graph and $a\neq b$ be two vertices such that $a$
  is the only neighbor of $b$. Then, as a polynomial equation, we have
  \begin{equation}
    \label{eq:deg1red}
    I(G,x_a,x_b)=(1+x_b)I(\vdel G b, x_a/(1+x_b)),
  \end{equation}
  where $I(G,y,z)$ denotes $I(G;\gv x)$ with $x_a=y$ and $x_b=z$ and
  $I(\vdel G b, z)$ denotes $I(\vdel G b; \gv y)$ with $y_a=z$ and
  $y_v=x_v$ for all $v\in V\setminus \{a,b\}$. ($\vdel G b$ is defined
  on Page~\pageref{def:vdel}.)
\end{lem}
\begin{proof}
  Let $V'=V\setminus \{a,b\}$ and $i(A)=1$ if $A\subs V$ is a
  independent set in $G$ and $i(A)=0$ otherwise. We have
  \begin{align*}
    I(G,x_a,x_b) =& \sum_{A\subs
      V'}x_A\big(i(A)+x_ai(A\uni\{a\})+x_bi(A\uni\{b\})\big) \\
    =& \sum_{A\subs V'}x_A\big(i(A)+x_ai(A\uni\{a\})+x_bi(A)\big) \\
    =& \sum_{A\subs V'}x_A\big(i(A)(1+x_b)+x_ai(A\uni\{a\})\big),
  \end{align*}
  from which the claim follows.
\end{proof}

In other words, Lemma~\ref{lem:deg1red} states that a single leaf $b$
and its neighbor $a$ can be ``contracted'' by incorporating the weight
of $b$ into $a$. In a very similar way, two vertices with the same set
of neighbors can be contracted:

\begin{lem}
  \label{lem:clone}
  Let $G=(V,E)$ be a graph and $a,b\in V$ two vertices that have
  the same set of neighbors. Then
  \begin{equation}
    I(G;\gv x)=I(\vdel G b;\gv y),
  \end{equation}
  where $y_v=x_v$ for all $v\in V\setminus \{a,b\}$ and
  $y_a+1=(1+x_a)(1+x_b)$.
\end{lem}
\begin{proof}
  Similar to the proof of Lemma~\ref{lem:deg1red}.
\end{proof}

Consider the following special case of a $S$-clone.

\begin{defi}
  Let $S$ be the multiset that consists of $k$ times the number $0$
  and $G=(V,E)$ be a graph. Then we write $\tc[k] G$ to denote
  $G_S$. We call $\tc[k] G$ the \defexpr{$k$-clone} of $G$.
\end{defi}

Applying Lemma~\ref{lem:clone} repeatedly yields the following
statement. Observations of this kind have been used for a long time
\cite{valiant_enumeration, jerrum_valiant_vazirani, roth,
  interlace_hard}.

\begin{thm}
  \label{thm:cloning}
  Let $G=(V,E)$ be a graph. We have the polynomial identity
  \begin{align*}
    I(\tc[k] G;x)=&I(G;(1+x)^k-1). \qed
  \end{align*}
\end{thm}

Let us now use Lemma~\ref{lem:deg1red} to derive a formula that
describes how a path, attached to one vertex, influences the
independent set polynomial. Basically, we derive an explicit formula
from recursive application of \eqref{eq:deg1red} (cf.\ the proof of
the formula for the interlace polynomial of a path by Arratia et
al.\ \cite[Proposition 14]{arratia_two_var_interl}).

\begin{thm}
\label{thm:pathreduction}
Let $G=(V,E)$ be a graph and $a_0\in V$ a vertex. For a positive
integer $k$, let $\tr[k] G$ denote the graph $G$ with a path of length
$k$ added at $a_0$, i.e.\ $\tr[k] G = (V\uni \{a_1, \ldots, a_k\},
E\uni \{a_0a_1, a_1a_2, \ldots, a_{k-1}a_k\})$ with $a_1, \ldots, a_k$
being new vertices. Let $\gv x$ be a vertex labeling of $\tr[k] G$
with variables. Then the following polynomial equation holds:
\begin{equation}
  \label{eq:pathreduction}
  I(\tr[k] G; \gv x)=C_k I(G; \gv y),
\end{equation}
where $y_v=x_v$ for all $v\in V\setminus\{a_0, \ldots, a_k\}$,
$y_{a_0}=B_k/C_k$, and $B_0=x_k$, $C_0=1$ and, for $0\leq i<k$,
\begin{equation}
\label{eq:BiCi_rek}
\begin{pmatrix}
  B_{i+1} \\
  C_{i+1}
\end{pmatrix} =
M(x_{k-i-1})
\begin{pmatrix}
  B_i \\ C_i
\end{pmatrix},
\end{equation}
\begin{equation}
  \label{eq:Mx}
  M(x)= \begin{pmatrix}
  0 & x \\
  1 & 1
\end{pmatrix}.
\end{equation}

Let now $x\in \R$ be nondegenerate for path reduction, $\lambda_1,
\lambda_2$ be as in Definition~\ref{def:lambda}, $x_a=x$ for all $a\in
\{a_0, \ldots, a_k\}$ and $\lo^{k+2}\neq\lt^{k+2}$. Then
\eqref{eq:pathreduction} holds with $B_k$ and $C_k$ defined as in
\eqref{eq:Bk} and \eqref{eq:Ck}.
\end{thm}

\begin{proof}
Let us write $I(\tr[k] G; x_0, \ldots, x_k)$ for $I(\tr[k] G; \gv
x)$ where $x_{a_j}=x_j$, $0\leq j\leq k$. Let us argue that we defined
the $B_i$, $C_i$ such that, for $0\leq i\leq k$,
\begin{equation}
\label{eq:BiCi_sinn}
I(\tr[k] G; x_0, \ldots, x_k) = C_i I\left(\tr[k-i]G; x_0, \ldots,
x_{k-i-1},\frac {B_i} {C_i}\right).
\end{equation}
This is trivial for $i=0$. As we have
\begin{align*}
  1+\frac {B_i}{C_i} \left(1+x_{k-i-1}u^2\right) = & \frac {C_i + B_i\left(1+x_{k-i-1}u^2\right)}{C_i}
  = \frac {C_{i+1}}{C_i},
\end{align*}
we see that \eqref{eq:BiCi_sinn} holds for all $1\leq i\leq k$: Use
Lemma~\ref{lem:deg1red} in the following inductive step:
\begin{align*}
  I(\tr[k] G; x_0, \ldots, x_k) = &C_i I(\tr[k-i]G; x_0, \ldots, x_{k-i-1},\frac {B_i} {C_i}) \\
  =& C_i \frac{C_{i+1}}{C_i} I\left(\tr[k-i-1]G; x_0, \ldots, x_{k-i-2},
  \frac {x_{k-i-1}} {\frac{C_{i+1}}{C_i}}\right)\\
  =& C_{i+1}I\left(\tr[k-i-1]G; x_0, \ldots, x_{k-i-2},
    \frac {B_{i+1}} {C_{i+1}}\right).
\end{align*}
Thus, \eqref{eq:pathreduction} holds as a polynomial equality.

Let us now consider $x$ as a real number, $x>-1/4$. Matrix $M(x)$ in
\eqref{eq:Mx} can be diagonalized as $M(x)=SDS^{-1}$ with
\[S=
\begin{pmatrix}
  x & x \\ \lambda_1 & \lambda_2
\end{pmatrix},\quad
D =
\begin{pmatrix}
  \lambda_1 & 0 \\ 0 & \lambda_2
\end{pmatrix},\quad
S^{-1}=\frac 1 {x(\lambda_2-\lambda_1)} \begin{pmatrix}
  \lambda_2 & -x \\ -\lambda_1 & x
\end{pmatrix},\]
$\lambda_{1,2}$ as in \eqref{eq:lolt}. Now we substitute variable
$x_v$ by real number $x$ for all $v\in \{a_0, \ldots, a_k\}$ and
$M(x)$ by $SDS^{-1}$. This yields the statement of the theorem.
\end{proof}

Now we see that Theorem~\ref{thm:S_cloning} can be proved by repeated
application of Lemma~\ref{lem:clone} and
Theorem~\ref{thm:pathreduction}.

%%%%%%%%%%%%%%%%%%%%%%%%%%%%%%%%%%%%%%%%%%%%%%%%%%%%%%%%%%%%%%%%%%%%%%
\section{Interpolation via $S$-clones}
\label{sec:interpolation}

In this section, we give a reduction from evaluation of the
independent set polynomial at a fixed point $x\in \Q\setminus \{0\}$
to computation of the coefficients of the independent set
polynomial. Thus, given a graph $G$ with $n$ vertices, we would like
to interpolate $I(G;X)$, where $X$ is a variable. The degree of this
polynomial is at most $n$, thus it is sufficient to know $I(G;x_i)$
for $n+1$ different values $x_i$. Our approach is to modify $G$ in
$n+1$ different ways to obtain $n+1$ different graphs $G_0$, \ldots,
$G_n$. Then we evaluate $I(G_0;x)$, $I(G_1; x)$, \ldots,
$I(G_n;x)$. We will prove that $I(G_i;x)=p_i I(G;x_i)$ for $n+1$ easy
to compute $x_i$ and $p_i$, where $x_i\neq x_j$ for all $i\neq
j$. This will enable us to interpolate $I(G;X)$.

If the modified graphs $G_i$ are $c$ times larger than $G$, we lose a
factor of $c$ in the reduction, i.e.\ a $2^n$ running time lower bound
for evaluating the graph polynomial at $x$ implies only a $2^{n/c}$
lower bound for evaluation at the interpolated points. Thus, we can
not afford simple cloning (i.e.\ constructing $\tc[2]G$, $\tc[3]G$,
\ldots\ to use Theorem~\ref{thm:cloning}): To get enough points for
interpolation, we would have to evaluate the graph polynomials on
graphs of sizes $2n, 3n, \ldots, n^2$. To overcome this problem, we
transfer a technique of Dell, Husfeldt, and Wahlén
\cite{dell_husfeldt_wahlen}, which they developed for the Tutte
polynomial to establish a similar reduction: We clone every vertex
$O(\log n)$ times and use $n$ different ways to add paths of different
(but at most $O(\log^2 n)$) length at the different
clones. Eventually, this will lead to the following result:

\begin{thm}
\label{thm:expreduction}
  Let $x_0\in \Q$ such that $x_0$ is nondegenerate for path reduction
  and the independent set polynomial $I$ of every $n$-vertex graph $G$
  can be evaluated at $x_0$ in time $2^{o(n/\log^3n)}$.

  Then, for every $n$-vertex graph $G$, the $X$-coefficients of the
  independent set polynomial $I(G;X)$ can be computed in time
  $2^{o(n)}$. In particular, the independent set polynomial $I(G;x_1)$
  can be evaluated in this time for every $x_1\in \Q$.
\end{thm}

Using this theorem, we can prove our main result.
\begin{proof}[Proof of Theorem~\ref{thm:WISexponential}]
  For $x>-1/4$, the corollary follows from
  Theorem~\ref{thm:expreduction} and Theorem~\ref{thm:IS}.

  Let us now consider $x<-2$. Then we have $|1+x|>1$, which implies
  $(1+x)^2-1>0$. On input graph $G=(V,E)$, we have $I(\tc[2]
  G;x)=I(G;(1+x)^{2}-1)$ by Theorem~\ref{thm:cloning}. Graph $\tc[2]
  G$ has $2|V|$ vertices. This establishes a reduction from $\gpe I
  {(1+x)^{2}-1}$ to $\gpe I x$, where the instance size increases only
  by a constant factor. As $(1+x)^2-1>0$, we have reduced from an
  evaluation point where we have already proved the (conditional)
  lower bound of the lemma. Thus, the same bound, which is immune to
  constant factors in the input size, holds for $\gpe I x$.

  Let us consider $x\in (-2,0)\setminus \{-1\}$. We have $|x+1|<1$ and
  $|x+1|\neq 0$. In a similar way as we just used a $2$-clone, we can
  use the comb reduction \cite[Section 3.2]{interlace_hard}: Let $k$
  be a positive even integer such that $k>\frac
  {\log(-2x)}{\log|x+1|}$. Then we have $y:=\frac x {(1+x)^k}<-2$. On
  input graph $G=(V,E)$, we can construct $G_k$ as in the comb
  identity for the interlace polynomial \cite[Theorem
    3.5]{interlace_hard}, and we have $I(G_k;x)=(1+x)^kI(G;y)$. As $k$
  does not depend on $n=|V|$, $|V(G_k)|=O(|V|)$. Thus, we have reduced
  from $y<-2$, an evaluation point where we have already proved the
  lower bound, to evaluation at $x$.

  To handle $x=-2$ and $x=-1$, add cycles \cite[Theorem 3.7 and
    Proposition 3.8]{interlace_hard}.
\end{proof}

The rest of this section is devoted to the proof of
Theorem~\ref{thm:expreduction}, which is quite technical. The general
idea is similar to Dell et al.\ \cite[Lemma 4, Theorem
  3(ii)]{dell_husfeldt_wahlen}.

\begin{defi}
  Let $S$ be a set of numbers. Then we define $\Ss S=\sum_{s\in S} s$.
\end{defi}

\begin{rem}
  \label{rem:S_clone_size}
  The $S$-clone $G_S$ of a graph $G=(V,E)$ has $|V|(\Ss S + |S|)$
  vertices.
\end{rem}

\begin{lem}
\label{lem:niceSi}
  Assume that $x\in \R$ is nondegenerate for path reduction. Then
  there are sets $S_0, S_1, \ldots, S_n$ of positive integers,
  constructible in time $\poly(n)$, such that
  \begin{enumerate}
  \item $x(S_i)\neq x(S_j)$ for all $i\neq j$ and
    \label{lem:niceSi:separation}
  \item $\Ss {S_i} \in O(\log^{3} n)$ and $|S_i|\in O(\log n)$ for all
    $i$, $0\leq i \leq n$.
    \label{lem:niceSi:sum}
  \end{enumerate}
\end{lem}

\begin{proof}
  We use the notation from Theorem~\ref{thm:pathreduction} and assume
  $\lo>\lt$.

  As $\left| \frac \lo \lt \right|^k\to \infty$ for $k\to \infty$,
  there is a positive integer $s_0$ such that
  \[ \left (\frac \lo \lt\right)^{s} \not \in \left\{ \left(\frac \lt \lo\right)^2, \frac {\lt(x+\lt)}{\lo(x+\lo)}
  \right\} \quad \forall s\geq s_0.\]

  Thus, for every $i$, $0\leq i\leq n$, the following set fulfills the
  precondition on $S$ and $T$ in Lemma~\ref{lem:Xdiff}:
  \[ S_i=\{s_0+\Delta (2j+b_j^{(i)})\ |\ 0\leq j\leq \lfloor \log n\rfloor\},\]
  where $\Delta$ is a positive integer defined later, $\Delta \in
  \Theta(\log n)$, and $[b_{\lfloor \log n\rfloor}^{(i)}, \ldots,
    b_1^{(i)}, b_0^{(i)}]$ is the binary representation of $i$. Note
  that this construction is very similar to Dell et al.\ \cite[Lemma
    4]{dell_husfeldt_wahlen}. It is important that the elements in
  these sets have distance at least $\Delta$ from each other. The sets
  are $\poly(n)$ time constructible as $s_0$ does not depend on
  $n$. We have $\Ss {S_i}\leq (1+\log n) (s_0 +(1+2\log n)\Delta)$ and
  obviously $|S_i|\in O(\log n)$ for all $i$. Thus, the second
  statement of the lemma holds.

  To prove the first statement, we use Lemma~\ref{lem:Xdiff}. Let
  $1\leq i<j\leq n$ and $S=S_i\setminus S_j$, $T=S_j\setminus
  S_i$. Let $s_1$ be the smallest number in $S\uni T$ and $A_1=(S\uni
  T)\setminus \{s_1\}$.  For $f$ as in Lemma~\ref{lem:Xdiff}, let us
  prove that $|f(A_1)|>\sum_{\substack{A\subs S\uni T\\ A\neq A_1}}
  |f(A)|$, which yields the statement of Lemma~\ref{lem:niceSi}.

  Assume without loss of generality that $s_1\in S$. As $x$ is
  nondegenerate, $C_1:=\min\{1,|\lo|, |\lt|, |x+\lo|, |x|,
  |\lo-\lt|\}$ is a nonzero constant. As $|S|=|T|$, \[
  \begin{split}
    D(S,T,A_1) &=\lo^{|T|} (x+\lo)^{|S|-1}(x+\lt) - \lo^{|S|-1}\lt(x+\lo)^{|T|} \\
    &=\lo^{|S|-1}(x+\lo)^{|S|-1}\big(\lo(x+\lt)-\lt(x+\lo)\big)\\
    &=\lo^{|S|-1}(x+\lo)^{|S|-1}x(\lo-\lt),
  \end{split} \]
  and we have
  \begin{equation}
    \label{eq:highestterm}
    |f(A_1)|\geq |\lo|^{\Ss {S\uni T}-s_1}|\lt|^{s_1}C_1^{7|S|}.
  \end{equation}
  
  If $A=\emptyset$ or $A=S\uni T$, we have $D(A)=0$, which implies
  $f(A)=0$. For every $A\subs S\uni T$, $A\neq \emptyset$, $A\neq
  S\uni T$, $A\neq A_1$, we have $\Ss A \leq \Ss{S\uni
    T}-s_1-\Delta$. Thus,
  \begin{equation}
    \label{eq:otherterms}
    |f(A)|\leq |\lo|^{\Ss{S\uni T}-s_1-\Delta}|\lt|^{s_1+\Delta}C_2^{7|S|},
  \end{equation}
  where $C_2=2\max\{1,|\lo|, |\lt|,|x+\lo|,|x+\lt|\}$. There are less
  than $2^{\lfloor \log n\rfloor + 1}\leq 2n^2$ such $A$. Combining
  this with \eqref{eq:highestterm} and \eqref{eq:otherterms}, it
  follows that we have proved the lemma if we ensure
  \[ \Big|\frac{\lo}{\lt}\Big|^{\Delta}>\Big(\frac{C_2}{C_1}\Big)^{7|S|}2n^2.\]
  This holds if
  \[ \Delta >7\Big((\log n + 1)\log \frac{C_2}{C_1} + 2\log n + 1\Big)/\log \frac{\lo}{|\lt|}.
  \]
  As $C_1$, $C_2$, $\lo$, $\lt$ do not depend on $n$, we can choose
  $\Delta\in \Theta(\log n)$.
\end{proof}

\begin{lem}
\label{lem:Xdiff}
  Let $S$ and $T$ be two sets of positive integers. Let also $x\in \R$
  be nondegenerate for path reduction and, for all $s\in S\uni T$,
  \begin{align}
    \left(\frac \lo \lt \right)^{s+2} &\neq 1\quad \text{and} \\
    \left(\frac \lo \lt \right)^{s+1} &\neq \frac {x+\lt} {x+\lo},
    \label{eq:path_nonzero}
  \end{align}
  where $\lo, \lt$ are defined as in Theorem~\ref{thm:pathreduction}.
  Then we have $x(S)=x(T)$ iff
  \[\sum_{A\subs S\sd T} f(A) = 0,\]
where
\begin{eqnarray*}
  f(A)&=&\lo^{\Ss A}\lt^{\Ss{(S\sd T)\setminus A}}
  (-\lo)^{|A|} \lt^{|(S\sd T)\setminus A|} \cdot D(S\setminus T, T\setminus S, A),\\
  D(S,T,A) &=& c(S,T,A\cut S,A\cut T) - c(T,S,A\cut T, A\cut S),\\
  c(S,T,S_0, T_0)&=&\lo^{|T_0|}\lt^{|T\setminus T_0|}(x+\lo)^{|S_0|}(x+\lt)^{|S\setminus S_0|}.
\end{eqnarray*}
\end{lem}

\begin{proof}
  Let $\S=S\setminus T$ and $\T=T\setminus S$. We have $x(S)=x(T)$ iff
  $x(S)+1=x(T)+1$. Condition \eqref{eq:path_nonzero} ensures
  $1+\frac{B_s}{C_s}\neq 0$ for all $s\in S\uni T$. Thus, $x(S\cut
  T)+1\neq 0$, and $x(S)=x(T)$ iff $x(\S)+1=x(\T)+1$. This is
  equivalent to $Y(\S,\T)=Y(\T,\S)$, where $Y(S,T)=\prod_{s\in S}
  (C_s+B_s)\prod_{t\in T}C_t$. For sets of integers $M\subs N$, let us
  define
  \begin{align*}
    B(N,M)&=\lo^{\Ss M}(-\lo)^{|M|}\lt^{\Ss{N\setminus M}}\lt^{|N\setminus M|} \quad \text{and}
    \\
    C(N, M) &= \lo^{\Ss {M}}\lo^{2|M|}(-1)^{|M|}\lt^{\Ss {N\setminus
        M}} \lt^{2|N\setminus M|}.
  \end{align*}
Using this notation, it is
\[\begin{split}
Y(S,T)= &
\prod_{s\in S}(B_s+C_s)\prod_{t\in T}C_t \\
=&\sum_{S_0\subs S}\prod_{s\in S_0}B_s\prod_{s\in S\setminus S_0}C_s \prod_{t\in T}C_t\\
=&(\lt-\lo)^{-|S|-|T|}\sum_{S_0\subs S}x^{|S_0|}\sum_{S_1\subs S_0}B(S_0,S_1)
\sum_{S_2\subs S\setminus S_0}C(S\setminus S_0,S_2)\\
&\sum_{T_0\subs T}C(T,T_0).
\end{split}
\]
We want to collect the terms $\lo^{\Ss M}$ and $\lt^{\Ss {N\setminus
    M}}$ in one place. Thus, we change the order in which $S$ is split
into subsets $S_0, S_1, S_2$ (cf.\ Figure~\ref{fig:summation_sets})
such that we first choose $\Sot := S_1\uni S_2 \subs S$, then
$S_1\subs \Sot$ (which implies $S_2 = \Sot\setminus S_1$), and finally
$S_0$ as $S_1\subs S_0\subs S\setminus S_2$.
\begin{figure}
  \centering
  \begin{tikzpicture}
    \draw (0,0) rectangle (3,1.5);
    \draw (1.7,-0.4) -- (1.3,1.9) node[left=0.4cm] {$S_0$} node [right=0.2cm] {$S\setminus S_0$};
    \draw (0.6,0) arc (0:90:0.6cm) node [below=0.4cm,right=0.0cm] {$S_1$};
    \draw (3,0.8) arc (90:180:0.8cm) node [above=0.3cm,right=0.2cm] {$S_2$};
  \end{tikzpicture}
  \caption{Partition of $S$.}
  \label{fig:summation_sets}
\end{figure}
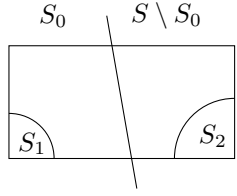
Now we can write
\begin{equation}
\label{eq:Yfinal}
\begin{split}
  Y(S,T)=(\lt-\lo)^{-|S|-|T|}\sum_{\Sot\subs S}\sum_{T_0 \subs
    T}&\lo^{\Ss \Sot + \Ss{T_0}}\lt^{\Ss {S\setminus
      \Sot}+\Ss{T\setminus
      T_0}}\\
  &(-\lo)^{|\Sot|+|T_0|} \lt^{|S\setminus\Sot|+|T\setminus T_0|}\\
  &c(S,T,\Sot,T_0),
\end{split}  
\end{equation}
where
\[
\begin{split}
  c(S,T,\Sot, T_0)=
  \lo^{|T_0|}\lt^{|T\setminus T_0|}
  \sum_{\substack{(S_1,S_2)\\S_1\dot \uni S_2=\Sot}}
  \sum_{\substack{S_0\\S_1\subs S_0\subs S\setminus S_2}}
  x^{|S_0|} \lo^{|S_2|}\lt^{|(S\setminus S_0)\setminus S_2|}.
\end{split}
\]
Note that \eqref{eq:Yfinal} as symmetric in $S$ and $T$, except for
the term $c(S,T,\Sot,T_0)$. Let us analyze this non-symmetrical
term. We write $S_0=S_1\dot \uni \Sz$.
\[
\begin{split}
  c(S,T,\Sot,T_0)= &\lo^{|T_0|}\lt^{|T\setminus T_0|}\sum_{S_1\dot \uni S_2=\Sot}\lo^{|S_2|}
  x^{|S_1|}
  \sum_{\Sz\subs S\setminus \Sot} x^{|\Sz|}\lt^{|(S\setminus \Sot)\setminus \Sz|} \\
  =&\lo^{|T_0|}\lt^{|T\setminus T_0|}\sum_{S_1\dot \uni S_2=\Sot}\lo^{|S_2|}
  x^{|S_1|} (x+\lt)^{|S\setminus \Sot|} \\
  =&\lo^{|T_0|}\lt^{|T\setminus T_0|}(x+\lo)^{|\Sot|}(x+\lt)^{|S\setminus \Sot|}.
\end{split}
\]
This implies the statement of the lemma.
\end{proof}

\begin{proof}[Proof of Theorem~\ref{thm:expreduction}]
  On input $G=(V,E)$ with $|V|=n$, do the following. Construct
  $G_{S_0}$, $G_{S_1}$, \ldots, $G_{S_n}$ with $S_i$ from
  Lemma~\ref{lem:niceSi}. Every $G_{S_i}$ can be constructed in time
  polynomial in $\gs {G_{S_i}}$, which is $\poly(n)$ by
  Remark~\ref{rem:S_clone_size} and by
  condition~\ref{lem:niceSi:sum}.\ of Lemma~\ref{lem:niceSi}. Thus,
  the whole construction can be performed in time $\poly(n)$.

  Again by condition~\ref{lem:niceSi:sum}.\ of Lemma~\ref{lem:niceSi},
  there is some $c'>1$ such that all $G_{S_i}$ have $\leq c'n\log^3 n$
  vertices. Evaluate $I(G_{S_0};x)$, $I(G_{S_1};x)$, \ldots,
  $I(G_{S_n};x)$. By the assumption of the theorem, one such
  evaluation can be performed in time
  \[ 2^{c\frac {c'n\log^3 n}{(\log(c'n\log^3 n))^3}}=2^{\frac{cc'n\log^3n}{(\log c' + \log n + 3\log\log n)^3}}
  \leq 2^{\frac{cc'n\log^3n}{(\log n)^3}} = 2^{cc'n}\] for every
  $c>0$.

  Using Theorem~\ref{thm:S_cloning}, we can compute $I(G;x(S_0))$,
  $I(G;x(S_1)$, \ldots, $I(G;x(S_n))$ from the already computed
  $I(G_{S_i};x)$ in time $\poly(n)$.

  By condition~\ref{lem:niceSi:separation}.\ of
  Lemma~\ref{lem:niceSi}, the $n+1$ values $x(S_i)$ are pairwise
  distinct. As $I(G;X)$ is a polynomial of degree at most $n$ in $X$,
  this enables us to interpolate $I(G;X)$. The overall time needed is
  $\poly(n)2^{cc'n}\leq 2^{(cc'+\eps)n}$ for every $\eps>0$.
\end{proof}

%%%%%%%%%%%%%%%%%%%%%%%%%%%%%%%%%%%%%%%%%%%%%%%%%%%%%%%%%%%%%%%%%%%%%%
\section{Open Problems}

The most important open problem is to find a reduction that does not
lose the factor $\Theta(\log^3n)$ in the exponent of the running time.

Another interesting direction for further research are restricted
classes of graphs, for example graphs of bounded maximum degree or
regular graphs.

The independent set polynomial is a special case of the two-variable
interlace polynomial \cite{arratia_two_var_interl}. It would be
interesting to have an exponential time hardness result for this
polynomial as well. In this context, the following question arises: Is
the upper bound $\exp(O(\sqrt n))$ \cite{DBLP:conf/isaac/SekineIT95}
for evaluation of the Tutte polynomial on \emph{planar} graphs sharp?

\subsection*{Acknowledgments}

I would like to thank Raghavendra Rao and the anonymous referees for
helpful comments.

%%%%%%%%%%%%%%%%%%%%%%%%%%%%%%%%%%%%%%%%%%%%%%%%%%%%%%%%%%%%%%%%%%%%%%

\bibliographystyle{amsplain}
\bibliography{literatur}

\end{document}